\documentclass{article}

\usepackage{natbib}
\usepackage{hyperref}

\usepackage{amsmath, amssymb, amsfonts, mathtools, microtype}
\usepackage{amsthm}

\usepackage{breqn} 
\usepackage{xcolor}

\usepackage{tikz}
\usepackage{pgfplots}
\pgfplotsset{width=0.45\textwidth, compat=1.15}

\usepackage{tcolorbox}
\tcbuselibrary{skins,breakable}
\usetikzlibrary{shadings,shadows}

\newenvironment{myblock}[1]{%
    \tcolorbox[beamer,%
    noparskip,breakable,
    colback=black!10!white,colframe=black,%
    colbacklower=black!10!white,%
    title={#1}]}%
    {\endtcolorbox}

\theoremstyle{plain}

\newtheorem{prop}{Proposition}
\newtheorem{lem}{Lemma}
\newtheorem{defn}{Definition}
\newtheorem{assum}{Assumption}
\theoremstyle{remark}
\newtheorem{rem}{Remark}
\newtheorem{exmp}{Example}

\def \BR{\mathsf{BR}}
\def \R{\mathbb{R}}

\def \D{\mathrm{D}}
\def \p{\mathrm{p}}
\def \a{\mathrm{a}}

\DeclareMathOperator*{\argmax}{arg\,max}

\DeclarePairedDelimiterX\set[1]\{\}{%
\renewcommand\given{\SetSymbol[\delimsize]}
#1
}
\DeclarePairedDelimiter\of()
\newcommand\ha{\hat a}
\newcommand{\con}[1]{#1^\star}

\newcommand{\va}{v_\a} 
\newcommand{\cva}{V_\a} 
\newcommand{\vp}{v_\p} 
\newcommand\eqdef\coloneqq
\newcommand\ol\overline 
\providecommand\given{}
\DeclarePairedDelimiterXPP\Ex[1]{\mathbb{E}}[]{}{
\renewcommand\given{\nonscript\:\delimsize\vert\nonscript\:\mathopen{}}
#1}

\providecommand{\keywords}[1]{{\textit{Index terms---}} #1}
\usepackage{authblk}

\title{Mitigating Information Asymmetry in Two-Stage Contracts with Non-Myopic Agents} 

\author[1]{Munther A. Dahleh} 
\author[1]{Thibaut Horel} 
\author[1,2]{M. Umar B. Niazi\thanks{This work is supported by the European Union’s Horizon Research and Innovation Programme under Marie Skłodowska-Curie grant agreement No. 101062523.}}

\affil[1]{Laboratory for Information and Decision Systems, Institute for Data, Systems, and Society, Massachusetts Institute of Technology, Cambridge MA 02139, USA (e-mail: \{\texttt{dahleh}, \texttt{thibauth}, \texttt{niazi}\}\texttt{@mit.edu}).}
\affil[2]{Division of Decision and Control Systems, Digital Futures, KTH Royal Institute of Technology, SE-100 44 Stockholm, Sweden.}

\date{}
\begin{document}
\maketitle

\begin{abstract}
	We consider a Stackelberg game in which a principal (she) establishes a
	two-stage contract with a non-myopic agent (he) whose type is unknown. The
	contract takes the form of an incentive function mapping the agent's
	first-stage action to his second-stage incentive. While the first-stage
	action reveals the agent's type under truthful play, a non-myopic agent could
	benefit from portraying a false type in the first stage to obtain a larger
	incentive in the second stage. The challenge is thus for the principal to
	design the incentive function so as to induce truthful play.
	We show that this is only possible with a constant,
	non-reactive incentive functions when the type space is continuous, whereas
	it can be achieved with reactive functions for discrete types.
    Additionally, we show that introducing an adjustment mechanism that penalizes inconsistent behavior across both stages allows the principal to design more flexible incentive functions.
\end{abstract}



\section{Introduction}\label{sec:intro}

The rise of digital technologies has turned human agents into active
participants within complex sociotechnical systems such as transportation
networks and power grids. Optimizing and controlling such systems require
planners to design incentives that take into account the interactions among
agents and their strategic behavior. 
For instance, \cite{brown2017} propose incentive mechanisms for routing drivers in a traffic network to reduce congestion, while \cite{niazi2023} design incentives to promote eco-driving in urban transportation to reduce emissions. The incentive mechanism proposed by \cite{satchidanandan2022} efficiently integrates electric vehicles into power grids as mobile battery storage units.
For a more in-depth exploration of incentive design applications in control, refer to \cite{chremos2024}.

When the principal (she) and the agent (he) interact repeatedly, integrating
learning algorithms with incentive mechanisms becomes key to optimizing system
performance in strategic environments. The InsurTech industry
\citep{holzapfel2023} is a good example of such interactions. During a
monitoring phase, the principal (the insurer) collects data to learn the agent's
behavioral patterns based on which the insurance plan is established for the
deployment stage. For instance, through InsurTech, the principal can
incentivize the agent to eco-drive during his daily commute or devise
an insurance policy based on the safety of the agent's driving behavior.

In this paper, we explore the challenges involved in learning when a principal
interacts with a strategic, non-myopic agent in a two-stage contract modeled as
a Stackelberg game. The principal makes decisions first, knowing that the agent
will observe her decisions and respond strategically. Since a non-myopic agent
anticipates his actions' consequences, the challenge is for the principal to
predict and influence his behavior effectively under adverse selection.

Recent advancements in learning-based incentive mechanisms focus on how a principal can leverage learning algorithms to align agent behavior with desired outcomes through incentives. A key challenge lies in addressing information asymmetry, where the principal lacks complete knowledge of agent preferences. When the agent is myopic, \cite{ratliff2020} propose an adaptive, model-based framework that integrates learning and control techniques. 
In a dynamic setting with an evolving system state, \cite{lauffer2023} propose an online learning algorithm for the principal, enabling her to learn and adjust incentives with unknown agent preferences. 
\cite{hutchinson2024} present a learning-based pricing mechanism for safety-critical networks, ensuring safety constraints while estimating user price responses over time.
Another development involves \cite{guruganesh2024} who extend algorithmic contract theory to such scenarios and propose dynamic contracts to optimize the principal's outcome, even when the agent uses a no-regret learning algorithm in response. 

The case of non-myopic agents has been comparatively less studied and is
particularly challenging in learning-based incentive design. Indeed, such
agents can manipulate the learning process by acting strategically, and
sometimes deceptively. For example, a driver might adopt an
uncharacteristically safe driving style during the monitoring phase to secure a
lower insurance premium in the deployment stage of a usage-based insurance
plan. Similarly, an agent might misrepresent himself as someone who
prioritizes travel time, thus manipulating the principal into offering a higher
reward for adopting eco-driving \citep{niazi2023}.

\cite{gerardi2020} focus on long-term contracts and examine strategies employed by the agent to broker increasingly favorable terms from the principal. Moreover, when the agent is sufficiently non-myopic, the screening mechanism proposed by \cite{gerardi2020} becomes infeasible, leading to inefficiencies in the allocation of incentives and revealing challenges due to dynamic adverse selection. \cite{birmpas2020} demonstrates this vulnerability and proposes methods for the principal to identify the best commitment strategy despite such manipulation.

When the agent also uses a learning method to tip the outcome in his favor, \cite{lin2024} show that the principal can achieve near-optimal payoffs when the agent uses a specific no-regret learning model, but that this performance can decline significantly against agents with no-swap-regret learning models. 
%
%
An important work by \cite{haghtalab2022} addresses various aspects of learning in Stackelberg games with non-myopic agents that discount future utilities. They explore applications like security games and strategic classification, and propose a bandit algorithm that reacts to the agent's past actions with a sufficiently large delay. 


In this paper, we explore incentive mechanisms that induce truthful play, meaning that a non-myopic agent finds it most beneficial to play according to his true type in both stages of the game. Our analysis in Section~\ref{sec:hidden} reveals that, when the mechanism ignores the agent's second-stage action, only constant incentive functions (independent of the agent's first-stage action in the first) can induce truthful behavior for a continuous type space. However, for a discrete type space, we show the existence of monotonic step functions as effective incentives. These functions reward the agent in proportion to his action or ``effort'' level, and induce truthful play when the jumps between reward levels are sufficiently small.

The fact that only constant functions induce truthful play for continuous types
is due to the combination of two kinds of information asymmetry in the
contract: adverse selection in the first stage (the principal does not know the
agent's type) and moral hazard in the second stage (the mechanism ignores the
second-stage action). Even though the first-stage action fully reveals the
agent's type when he plays truthfully, the principal cannot use the learned
type to verify that the second-stage action is consistent with it. This creates
an uneven playing field which advantages the agent due to the lack of
transparency in his actions.

The principal faces two challenges in designing such contracts. First-stage
information asymmetry, also known as adverse selection, means that she does not
know the agent's true type. Therefore, the incentive function depends on the
agent's first-stage action that immediately reveals his type if and only if he
plays truthfully. Second, if the mechanism does not take into account the
agent's second-stage action, the principal faces moral hazard. This combination
makes it difficult to induce truthful play, as the principal cannot verify the
agent's type learned in the first stage by comparing it with the second-stage
action.

To address these challenges, we introduce an additional adjustment mechanism in
Section~\ref{sec:observable}, which penalizes inconsistencies in the agent's
actions while also properly incentivizing truthful play when the agent is
consistent, that is, when his second-stage action matches the type portrayed by
the first-stage action. We prove that introducing such an adjustment mechanism can induce truthful play.

\section{Problem Definition}\label{sec:single}

We consider a contracting scenario between two players: the principal (referred to as \textit{she}) and
the agent (referred to as \textit{he}). After establishing notation and preliminaries, this section defines the Stackelberg game that models a single-stage principal-agent contract. We then analyze the agent's optimal response given the principal's chosen incentive under complete information. Subsequently, we introduce the two-stage contract under incomplete information and define truthful play under this contract, assuming the agent is non-myopic and his private information (type) is unknown to the principal.

\subsection{Notations and Preliminaries}\label{sec:prelims}

We denote by $\ol\R=\R\cup\set{-\infty,+\infty}$ the extended real line, and
for $A\subseteq\ol\R$, $A_+$ and $A_+^*$ denote respectively the sets of
non-negative and positive numbers in $A$.

Let $A\subseteq\R$ be an open set, then $\phi'$ denotes the derivative of the
function $\phi:A\to\R$ (when it exists), $\partial\phi$ denotes its
subdifferential, and $\D_i\pi$ denotes the partial derivative of the function
$\pi:A^2\to\R$ with respect to the $i$th argument (when it exists).

For a function $\phi:A\to\ol\R$ defined over a subset $A\subseteq\R$,
$\con\phi:\R\to\ol\R$ denotes its \emph{convex conjugate} (a.k.a.\
Fenchel-Legendre transform):
\begin{equation}
    \label{eq:conjugate}
	\con\phi(u) \eqdef \sup_{a\in A} \big[ua-\phi(a)\big].
\end{equation}
An immediate consequence of the definition is the Fenchel–Young inequality:
\begin{displaymath}
	\forall a\in A,\;\forall u\in U,\;\phi(a)+\con\phi(u)\geq au
\end{displaymath}
with equality iff $u\in\partial\phi(a)$.

The \emph{Bregman divergence} of a differentiable function $\phi:A\to\R$ is $D_\phi:A^2\to\R$ defined for $(x,y)\in A^2$ by
\begin{displaymath}
	D_\phi(x,y) = \phi(x)-\phi(y)-\phi'(y)\cdot (y-x).
\end{displaymath}
Note that $D_\phi$ vanishes on the diagonal $y=x$ and if $\phi'$ is continuous
(e.g.\ when $\phi$ is convex) then the same holds for $D_\phi$'s partial
derivatives. Finally, $D_\phi$ is non-negative when $\phi$ is convex.

\subsection{Single-stage Stackelberg Game}
We denote by $U\subseteq \R$ and $A\subseteq\R$ the action sets of the
principal and the agent, respectively. The principal's utility function
$\vp:U\times A\to\R$ is given by
\begin{equation}\label{eq:principal-utility}
    \vp(u,a) = \rho(a) - u a
\end{equation}
for some function $\rho:A \to \R$. The agent's utility $\va:U\times A\to\R$ is parametrized by his
type $\theta$ taking values in a type space $\Theta\subseteq\R$ and is given by
\begin{equation}
    \label{eq:agent-utility}
    \va(u,a;\theta) = u a - \theta \phi(a)
\end{equation}
for some function $\phi:A\to\R$.

We interpret the agent's action $a$ as an “effort” level (broadly conceived)
and the principal's action $u$ as a (marginal) monetary incentive for the agent
to exert a desired level of effort. With this interpretation, $\phi(a)$ can be
thought of as the cost incurred by the agent for adopting the effort level $a$,
while $\rho(a)$ is the corresponding benefit to the principal. We assume that
the principal moves first by announcing the incentive $u\in U$, and letting the
agent adapt his effort level accordingly.

\begin{exmp}
    The agent's action $a$ can be interpreted as his eco-driving level. This
	reflects the strategies he chooses to minimize emissions or fuel
	consumption. For instance, smoother acceleration and maintaining a
	consistent speed could be high eco-driving levels. The travel time for the
	agent, denoted by the function $\phi(a)$, depends on his eco-driving level and traffic conditions. Generally, more eco-friendly driving might lead to slightly longer travel times. To incentivize eco-driving, the principal offers a reward, $ua$, based on the agent's chosen level, $a$. Finally, the function $\rho(a)$ represents the reduction in emissions achieved by the agent's eco-driving level, $a$. 
    \hfill $\diamond$
\end{exmp}


Throughout the paper, we make the following
assumptions. 

\begin{assum}
\label{assume:main}
$\Theta\subseteq\R_+^*$, $U\in\set{\R,\R_+^*}$, and $\phi$ is differentiable and strictly convex with $\phi'(A)=U$. 
\hfill $\diamond$
\end{assum}

Assumption~\ref{assume:main} guarantees that the agent's best-response to an incentive
$u$ is well-defined as stated below.


\begin{lem}\label{lem:br}
	For all $\theta\in\Theta$ and $u\in U$, the agent has a unique
	best-response in $A$:
	\begin{equation}\label{eq:br}
		a_\theta(u)\eqdef (\phi')^{-1}\of*{\frac u\theta}.
	\end{equation}
	For all $\theta\in\Theta$, the function $a_\theta:U\to A$ is
	a homeomorphism with inverse $a\mapsto \theta\phi'(a)$.
	Finally, the agent's utility when best-responding is given by
	\begin{equation}\label{eq:br-v}
		\va\big(u,a_\theta(u);\theta\big) = \theta\con\phi\of*{\frac u\theta}
	\end{equation}
    where $\con\phi$ is the convex conjugate of $\phi$. \hfill $\diamond$
\end{lem}

If the agent's type $\theta$ is
known by the principal, the natural solution concept is a Stackelberg
equilibrium: the principal chooses $u\in U$ so as to maximize her utility
after the agent best-responds. In other words, the principal wants to
solve
\begin{equation}\label{eq:stack}
	\sup_{u\in U} \big[\rho\big(a_\theta(u)\big) - u a_\theta(u)\big].
\end{equation}

\begin{lem}\label{lem:stack}
	For all $\theta\in\Theta$, the single stage game has a complete
	information Stackelberg equilibrium iff 
    \begin{equation}\label{eq:principal-prob}
		\sup_{a\in A}\big[\rho(a)-\theta a\phi'(a)\big].
	\end{equation}
	is attained. When this condition is satisfied, $(u_e,a_e)\in U\times A$ is a
	Stackelberg equilibrium iff $a_e$ is a solution to \eqref{eq:principal-prob}
	and $u_e=\theta\phi'(a_e)$. \hfill $\diamond$
\end{lem}


Lemma~\ref{lem:stack} reveals an important feature of the principal's
decision problem that will be useful in the subsequent sections. Instead of
thinking of the principal's action as choosing an incentive, we can
equivalently think of it as choosing which action the agent should take. Then,
the principal can “reverse-engineer” the incentive that would induce this
action with the help of the inverse best response function $a_\theta^{-1}$.

\begin{rem}
	Observe that the supremum in \eqref{eq:principal-prob} is reached under
	rather mild assumptions. For example, it is easy to check that it suffices
	to have $\rho$ continuous with $a\mapsto \rho(a)/a$ bounded.
    \hfill $\diamond$
\end{rem}

\begin{rem}
    When the principal does not know the agent's
	type $\theta$ but has a prior distribution with density $p_\theta$ over $\Theta$, the best she can do is to choose
    $u$ so as to maximize her expected utility:
    \begin{equation}\label{eq:ex-ante-opt}
	\sup_{u\in U} \Ex*{\rho\big(a_\theta(u)\big)-ua_\theta(u)}
    \end{equation}
    where the expectation is over $p_\theta$. \hfill $\diamond$
\end{rem}


\subsection{Two-stage Contract with Unknown Type}
\label{subsec:contract}

We now extend the single-stage Stackelberg game studied in the previous section to a
two-stage contract under incomplete information, where the agent's type is unknown to the principal. The timing of the contract is as follows.

\begin{myblock}{Two-stage contract}
\begin{enumerate}
    \leftskip-0.2in
    \item At stage $t=1$:
    \begin{itemize}
        \leftskip-0.32in
        \item The principal plays $u_1\in U$ and commits to an incentive function $u_2: A \to U$ for the second stage.
        \item The agent plays $a_1 \in A$, which is observed by the principal.
    \end{itemize}

    \item At stage $t=2$:
    \begin{itemize}
        \leftskip-0.32in
        \item The incentive $u_2(a_1) \in U$ is realized according to the
			principal's commitment in the first stage.
        \item The agent plays $a_2\in A$, which may not be observed by the principal.
    \end{itemize}
\end{enumerate}
\end{myblock}

We interpret the first stage as a \emph{learning phase}, in which the principal
attempts to learn the type of the agent by observing his action $a_1$ and the second
stage as a \emph{deployment stage}, which exploits the knowledge gathered in the
first stage. The agent is non-myopic and plays in a way that maximizes his
\emph{cumulative} utility:
\begin{equation*}
	\va(u_1,a_1; \theta) + \va(u_2(a_1), a_2; \theta) 
	= u_1 a_1 - \theta \phi(a_1) + u_2(a_1) a_2 - \theta \phi(a_2).
\end{equation*}
The crucial feature of this game is that the second stage incentive,
$u_2(a_1)$, only depends on the first stage action. This models situations
in which the second stage action cannot be observed by the principal, or where,
for practical reasons, the incentive must be fixed at the beginning of the
second stage. The principal thus faces \emph{moral hazard}: the agent is
not accountable at the second stage and will always choose the action $a_2$
that maximizes his second stage utility. In other words, the agent plays
$a_2=a_\theta\big(u_2(a_1)\big)$, and using \eqref{eq:br-v}, his cumulative
utility becomes
\begin{equation}\label{eq:cumulative-utility}
	\cva(a_1) \eqdef u_1 a_1 -\theta\phi(a_1)+\theta\con\phi\of*{\frac{u_2(a_1)}{\theta}}.
\end{equation}

Moreover, the principal announces and commits to the incentive function $u_2$ at the
beginning of the game. In other words, even though the principal reacts to what
is observed in the first stage, the functional form of this reaction is known
by the agent and can thus be anticipated.

When the agent plays myopically in the first stage, $a_1=a_\theta(u_1)$, he reveals his type to the principal whenever
$\phi'\big(a_\theta(u_1)\big)\neq 0$. Indeed, by~\eqref{eq:br}, we have $\theta
= u_1 / \phi'\big(a_\theta(u_1)\big)$ in this case. However, the agent might
prefer to deviate from the myopic best response in the first stage, in order
not to fully reveal his type to the principal—or equivalently, pretend to be
of a different type—and secure himself a larger incentive at the second
stage.  In other words, the principal faces \emph{adverse selection} in the first
stage.

It is natural to ask whether it is possible for the principal to design $u_2$
so as to counteract the agent's adverse selection. This leads to the following
definition.

\begin{defn}
    We say that the incentive function $u_2:A\to U$ induces \emph{truthful} play if for
	every type $\theta\in\Theta$,
    \begin{equation}
    \label{eq:truthfulness-def}
    \cva\big(a_\theta(u_1)\big)\geq \cva(\hat a_1), \quad \text{for every}~ \hat a_1\in A. \quad \diamond
    \end{equation}
\end{defn}

Truthful play means that it is in the agent's best interest to play a
best-response \eqref{eq:br} at the first stage. This is useful because it
allows the principal to learn the agent's type $\theta$ from his action $a_1$.
If the agent always play truthfully, ($a_1=a_\theta(u_1)$), then the principal
would naturally designs the incentive function $u_2$ such that, for all
$\theta\in\Theta$, $u_2(a_\theta(u_1))$ solves \eqref{eq:stack}. However, if
the agent does not play truthfully, he can exploit the incentive function $u_2$
by causing the principal to learn a false type $\hat \theta$. Therefore, the
principal needs a more careful approach beyond just solving \eqref{eq:stack}.
In fact, she must seek a functional form of the incentive $u_2$ that induces
truthful play while maximizing her expected utility as in
\eqref{eq:ex-ante-opt}.

\begin{exmp}
    Consider a usage-based or pay-as-you-drive insurance. In this system, an agent's driving behavior during a monitoring phase (action denoted by $a_1$) determines their safety level.
    Safer driving (higher $a_1$) translates to a lower insurance premium in the following phase. The incentive for safe driving is essentially the reduced insurance cost, represented by $u_2(a_1)$. However, there can be inconveniences associated with driving more cautiously, such as increased travel time. This inconvenience is captured by the function $\phi:A\to \R$.  Additionally, some factors, like poor training or physical limitations, can make safe driving more challenging for an agent. These challenges are represented by the parameter $\theta$. A higher $\theta$ indicates greater difficulty in driving safely.

    Here's the key issue: according to \eqref{eq:br}, when $\theta$ is high, truthful safe driving, $a_1=a_\theta(u_1)$, results in a lower safety level and a higher premium.  A strategic non-myopic agent might exploit this by pretending to be extra safe (i.e., he chooses $a_1 \gg a_\theta(u_1)$) during monitoring phase to get a lower premium. However, there's no guarantee this safe behavior will continue in the actual insurance phase.
    Therefore, the insurance provider (principal) needs to design the incentive function $u_2:A\to U$ to encourage truthful safe driving in the monitoring phase so that an appropriate premium can be set for the agent in the deployment stage. This ensures fair pricing regardless of the agent's attempts to manipulate the system.
    \hfill $\diamond$
\end{exmp}



\section{Two-stage Contract with Moral Hazard in the Second Stage}
\label{sec:hidden}

In this section, we investigate the question of existence of
non-trivial incentive functions $u_2:A\to U$ that induce truthful play $a_1=a_\theta(u_1)$. By non-trivial we mean that the incentive $u_2(a_1)$ varies with $a_1\in A$. We obtain very contrasted answers depending on
whether the type space is continuous or discrete. 




\subsection{Inducing Truthful Play under Continuous Types}\label{sec:continuous}

When the type space $\Theta$ is continuous, the only incentive function $u_2$ that induces truthful play is trivial.

\begin{prop}
    \label{prop:impossibility}
    Assume $\Theta=\R_+^*$ and $U=\R_+^*$. Then, the incentive function $u_2: A\rightarrow U$ induces truthful play iff it is constant. \hfill $\diamond$
\end{prop}

In other words, the incentive function $u_2$ is completely non-reactive to the agent's action $a_1$ in the first stage.
By being non-reactive, the principal guarantees the agent that she will not
learn his type $\theta$, or at least not choose the incentive of the second
stage based on his type. As a result, the agent does not gain anything by
misrepresenting his type in the first stage. But as a result, the principal
loses the ability to adjust the agent's incentive in the second stage based on
the first-stage action and thus cannot exploit the type revealed by truthful
play.

Therein lies the dilemma: to incentivize truthfulness, the principal is forced to commit to a non-reactive incentive function, which makes the two-stage contract trivial. On the other hand, if the incentive function $u_2$ varies with the agent's first-stage action $a_1$, then it is always in the best interest of a non-myopic agent to play untruthfully.

\begin{proof}[Proof of Proposition~\ref{prop:impossibility}]
	Consider $\theta\in\Theta$ and $a=a_\theta(u_1)$. Using
	\eqref{eq:cumulative-utility}, inducing truthful play is equivalent to requiring that for all $\ha\in A$
	\begin{displaymath}
		u_1a-\theta\phi(a)+\theta\con\phi\of*{\frac{u_2(a)}\theta}
		\geq u_1\ha-\theta\phi(\ha)+\theta\con\phi\of*{\frac{u_2(\ha)}\theta}.
	\end{displaymath}
	Dividing by $\theta$ and rearranging
	\begin{displaymath}
		\con\phi\of*{\frac{u_2(\ha)}\theta}
		-\con\phi\of*{\frac{u_2(a)}\theta}
		\leq \phi(\ha)-\phi(a)-\frac{u_1}\theta(\ha-a).
	\end{displaymath}
	Finally, we eliminate $\theta$ using that
	$\phi'(a)=\frac{u_1}\theta$ due to \eqref{eq:br} and recognize a Bregman
	divergence on the right-hand side:
	\begin{equation}\label{eq:truthfulness}
		\con\phi\of*{\frac{u_2(\ha)\phi'(a)}{u_1}}
		-\con\phi\of*{\frac{u_2(a)\phi'(a)}{u_1}}
		\leq D_\phi(\ha,a).
	\end{equation}
	In other words, inducing truthful play is equivalent to
	\eqref{eq:truthfulness} for $a=a_\theta(u_1)$ and all $\ha\in A$, where $a_\theta(u_1)$ is given in \eqref{eq:br}. But as
	$\theta$ ranges over $\Theta$, $\frac{u_1}\theta$ ranges over $U$ (due to
	$\Theta=U=\R_+^*$), hence
	\begin{displaymath}
		\set*{a_\theta(u_1)\given \theta\in\Theta} =(\phi')^{-1}(U) =A
	\end{displaymath}
	where the last equality follows from Assumption~\ref{assume:main}.
	Therefore, inducing truthful play with respect to every $\theta\in\Theta$ is equivalent to
	\eqref{eq:truthfulness} for every $(a,\ha)\in A^2$, where $a=a_\theta(u_1)$.

\vspace{\baselineskip}\noindent\emph{Sufficiency.} If $u_2$ is
constant, then the left-hand side in \eqref{eq:truthfulness} is $0$, and by
convexity of $\phi$, the divergence $D_\phi$ is non-negative hence
truthfulness is satisfied.

\vspace{\baselineskip}\noindent\emph{Necessity [proof sketch].} 
Assume that $u_2$
induces truthful play, equivalently that \eqref{eq:truthfulness} holds for all
$(a,\ha)\in A^2$. 
First, we establish that $u_2$ is continuous over $A$. This can be shown by rearranging \eqref{eq:truthfulness} and writing in terms of $u_2(\hat a)$. Then, by letting $\hat a \to a$, one has that $u_2$ is continuous at $a$ by the continuity of $\phi$, $\phi'$, and $\con\phi$.
Finally, by using the convexity of $\con\phi$, we consider $a$ such that $u_2(a)\neq 0$ and show that $u_2$ is differentiable at $a$ with $u_2'(a)=0$. 
\end{proof}

\subsection{Inducing Truthful Play under Discrete Types}\label{sec:discrete}

We saw that when the type space $\Theta$ is continuous, the only incentive function $u_2$ that induces truthful play is a trivial, constant function.
However, when the type space is discrete, there are non-trivial incentive functions that induce truthful play. These incentive functions can be characterized as step functions with a sufficiently small step size.

\begin{prop}\label{prop:truth-d}
	Assume that $\Theta=\set{\theta_L,\theta_H}$ with $\theta_L<\theta_H$ and
	consider the incentive function
	\begin{displaymath}
		u_2(a_1)=\begin{cases}
			u_L&\text{if } a_1\geq t\\
			u_H&\text{otherwise}
		\end{cases}
	\end{displaymath}
	with
	\begin{displaymath}
		\frac{u_1}{\theta_H}\leq
		\phi'(t)\leq\frac{u_1}{\theta_L}.
	\end{displaymath}
	Then, when $u_L>u_H$, the mechanism is truthful iff
	\begin{equation}\label{eq:truth-d-1}
		\con\phi\of*{\frac{u_L}{\theta_H}}
		-\con\phi\of*{\frac{u_H}{\theta_H}}
		\leq \phi(t)+\con\phi\of*{\frac{u_1}{\theta_H}} - \frac{u_1
		t}{\theta_H}.
	\end{equation}
	Otherwise, when $u_H>u_L$, the mechanism is truthful iff
	\begin{equation}\label{eq:truth-d-2}
		\con\phi\of*{\frac{u_H}{\theta_L}}
		-\con\phi\of*{\frac{u_L}{\theta_L}}
		\leq \phi(t)+\con\phi\of*{\frac{u_1}{\theta_L}} - \frac{u_1
		t}{\theta_L}.
	\end{equation}
\end{prop}
\begin{proof}
	We give a proof for the case $u_L>u_H$. The case $u_H>u_L$ is exactly identical
	after swapping $L$ for $H$.

	Observe that, if an agent of type $\theta_L$ plays truthfully in the first
	stage (that is, $\phi'(a_1) = u_1/\theta_L$) then he gets the highest possible
	incentive $u_L$ in the second stage and thus have no incentive to
	deviate. Hence, only an agent of type $\theta_H$ could benefit from playing
	untruthfully in the first stage in order to secure the incentive $u_L$ in
	the second stage.

	It is easy to see that for an agent of type $\theta_H$, the choice of
	action $a_1$ which maximizes his first stage utility while guaranteeing
	incentive $u_L$ at the second stage is precisely $t$. In order for this
	deviation not to be beneficial, we need
	\begin{displaymath}
		u_1t-\theta_H\phi(t)+\theta_H\con\phi\of*{\frac{u_L}{\theta_H}}
		\leq
		\theta_H\con\phi\of*{\frac{u_1}{\theta_H}}+\theta_H\con\phi\of*{\frac{u_H}{\theta_H}}.
	\end{displaymath}
	We obtained the desired condition after dividing the previous inequality
	by $\theta_H$ and rearranging. 
\end{proof}

\begin{rem}
	Note that by the Fenchel–Young inequality (cf.~Section~\ref{sec:prelims}) the
	upper-bounds in \eqref{eq:truth-d-1} and \eqref{eq:truth-d-2} are
	non-negative, and even positive as long as $\phi'(t)> u_1/\theta_H$ for
	\eqref{eq:truth-d-1} and $\phi'(t)< u_1/\theta_L$ for \eqref{eq:truth-d-2}.
	In other words, the characterization of truthfulness in Proposition~\ref{prop:truth-d}
	is nontrivial in that it allows for a range of values for $u_L$ and $u_H$
	with $u_L\neq u_H$.
	Furthermore, the most permissive setting of $t$ (that is, the one making
	the set of truthful mechanisms the largest) is achieved when
	$\phi'(t)=u_1/\theta_L$ for \eqref{eq:truth-d-1} and
	$\phi'(t)=u_1/\theta_H$ for \eqref{eq:truth-d-2}.
\hfill $\diamond$
\end{rem}

\section{Two-stage Contract with Adjustment Mechanism}\label{sec:observable}

When the principal cannot observe the agent's second-stage action, we saw in
Section~\ref{sec:hidden} that for a continuous type space,
the only incentive functions that induce truthful play are the constant ones.
Because the principal does not know the type $\theta$ of the agent, she faces adverse selection in the first stage. 
When she cannot observe the action of the agent in the second stage, 
she also faces moral hazard. As a result, she cannot validate the agent's type she learned in the first stage.
This gives the agent an unfair advantage over the principal
as it is impossible for her to observe any discrepancies
between the agent's actions in both stages. 

In this section, we aim to mitigate moral hazard by relaxing the assumption
that the principal cannot observe the second stage action. Specifically, in
addition to the second stage incentive function $u_2:A\to U$, we introduce an
\emph{adjustment} function $\pi:A^2\to\ol\R$ which depends on both actions and
can be thought of as an additional reward or penalty (depending on the sign of
$\pi$) collected at the end of the second stage. The principal commits to $\pi$
along with $u_2$ at the beginning of the game, and the timing of the modified
two-stage contract is as follows.

\begin{myblock}{Two-stage contract with adjustment}
\begin{enumerate}
    \leftskip-0.2in
    \item At stage $t=1$:
    \begin{itemize}
        \leftskip-0.32in
        \item The principal plays $u_1\in U$ and commits to an incentive $u_2: A \rightarrow U$ and an adjustment $\pi:A^2\to \ol \R$ for the second stage.
        \item The agent plays $a_1 \in A$, which is observed by the principal.
    \end{itemize}
    \item At stage $t=2$:
    \begin{itemize}
        \leftskip-0.32in
        \item The incentive $u_2(a_1)$ is realized according to the principal's commitment at the first stage. 
        \item The agent plays $a_2 \in A$, which is also observed by the principal.
        \item The adjustment $\pi(a_1,a_2)$ is realized according to the principal's commitment at the first stage.
    \end{itemize}
\end{enumerate}
\end{myblock}

In the modified contract above, the cumulative utility of a non-myopic
agent of type $\theta\in\Theta$ is given by
\begin{subequations}
    \label{eq:F_a-2}
    \begin{align}
        \cva(a_1,a_2;\theta) &= \va(u_1,a_1; \theta) + \va(u_2(a_1), a_2; \theta) - \pi(a_1,a_2) \\
        &= u_1 a_1 - \theta \phi(a_1) + u_2(a_1) a_2 - \theta \phi(a_2) - \pi(a_1,a_2).
    \end{align}    
\end{subequations}

\subsection{Inducing Truthful Play through Adjustment Function}

The goal of the principal is now to jointly design the functions $u_2$ and $\pi$
so as to guarantee that an agent of type $\theta$ maximizes his cumulative
utility by (i) playing a best-response to $u_1$ in the first stage:
$a_1=a_\theta(u_1)$, and (ii) playing a best-response to $u_2(a_1)$ in the
second stage: $a_2=a_\theta\big(u_2(a_1)\big)$. When this is the case, we say
that the pair $(u_2,\pi)$ \emph{induces truthful play}.

In case of truthful play by the agent, there is a form of \emph{consistency}
satisfied by the pair of actions $(a_1,a_2)$. Indeed, since the agent plays
according to his true type $\theta$ at both stages of the game, we have by
Lemma~\ref{lem:br} that
\begin{displaymath}
	\frac{\phi'(a_1)}{u_1} = \frac{\phi'(a_2)}{u_2(a_1)}
\end{displaymath}
which allows us to write the second stage action as a function of the first
stage action
\begin{displaymath}
a_2= c(a_1) \coloneq (\phi')^{-1}\of*{\frac{u_2(a_1)}{u_1}\phi'(a_1)}.
\end{displaymath}
We will henceforth refer to the graph of the function $c$ as the
\emph{consistency curve}.


    

The following proposition shows that any choice of (differentiable) incentive
function $u_2$ can induce truthful play when paired with a suitable adjustment
function $\pi$.

\begin{prop}\label{prop:truth-gen}
	Assume $\Theta=\R_+^*$ and $U=\R_+^*$. Let $u_2:A\to U$ be a differentiable
	function and consider an adjustment function of the form 
    \begin{equation}
        \pi(a_1,a_2) = \left\{\begin{array}{cl}
        f(a_1) & \text{if } a_2 = c(a_1) \\
        +\infty & \text{otherwise}
    \end{array}\right.
    \end{equation}
	for some function $f:A\to\R$.

	Then, the pair $(u_2,\pi)$ induces truthful play iff
	$f$ is differentiable with
    \begin{equation}
        \label{eq:f'}
        f'(a_1) = g'(a_1) - u_1 \frac{(\phi\circ c)'(a_1)}{\phi'(a_1)}.
    \end{equation}
\end{prop}
\begin{proof}
	We consider a deviation $(\ha_1,\ha_2)$ by the agent. Since
	$\pi(\ha_1,\ha_2)=+\infty$ if $\ha_2\neq c(\ha_1)$, we can assume without
	loss generality that $\ha_2=c(\ha_1)$. Then inducing truthful play \eqref{eq:truthfulness-def} is
	equivalent to
    \begin{equation}
        \label{eq:truthful-def-equiv}
        f(\hat a_1) - f(a_1) \geq g(\hat a_1) - g(a_1)
        - \frac{u_1}{\phi'(a_1)} \Bigl[D_\phi(\hat a_1, a_1) + \phi\circ c(\hat a_1) - \phi\circ c(a_1) \Bigr]
    \end{equation}
    where $a_1=a_\theta(u_1)$ and $\hat a_1\in A$.
    
    \emph{Sufficiency.}
    By integrating \eqref{eq:f'} from $a_1=a_\theta(u_1)$ to $\hat a_1\in A$, where $\hat a_1\geq a_1$ (the proof for $\hat a_1\leq a_1$ is similar), we obtain
    \begin{displaymath}
        f(\hat a_1) - f(a_1) = g(\hat a_1) - g(a_1) -u_1 \int_{a_1}^{\hat a_1} \frac{(\phi\circ c)'(z)}{\phi'(z)}dz.
    \end{displaymath}
    Since $\phi$ is convex, $\phi'(z)\geq \phi(a_1)$ and $D_\phi(\hat a_1,a_1)\geq 0$. Therefore, we conclude that \eqref{eq:truthful-def-equiv} is satisfied.

    \emph{Necessity [proof sketch].}
    Let \eqref{eq:truthful-def-equiv} hold. Then, by approaching $\hat a_1\downarrow a_1$ and $\hat a_1 \uparrow a_1$, one obtains \eqref{eq:f'}.
    
\end{proof}

\subsection{Consistency-inducing Penalty is Extreme}

For a pair $(u_2,\pi)$ that induces truthful play as in
Proposition~\ref{prop:truth-gen}, we see that the function $\pi$ plays two
distinct roles: (i) it penalizes inconsistent pairs of actions ($a_2\neq c(a_1)$), and
(ii) it properly compensates pairs of actions along the consistency curve.

In this section, we focus on the first of these two roles and consider
adjustment functions which \emph{only} penalize inconsistent play, but vanish
on the consistency curve. This leads to the following definition.

\begin{defn}
    \label{def:penalize-inconsistency}
	We say that the adjustment function $\pi:A^2\to \ol \R$ \emph{penalizes inconsistency} if it is
	non-negative and if
	\begin{displaymath}
		\pi(a_1,a_2)=0 \iff 
	a_2 = c(a_1).
	\end{displaymath}
Equivalently, the agent incurs no penalty when playing consistently, but
incurs a positive penalty otherwise. \hfill $\diamond$
\end{defn}

Let's say the agent plays $a_1=a_{\hat \theta}(u_1)$ in the first stage, for some $\hat \theta\in \Theta$ that may be different from his true type $\theta$. Then, $a_2=c(a_1)$ is the action the principal expects the agent to play in the second stage. If the agent's action in the second stage is inconsistent, i.e., $a_2\neq c(a_1)$, he is penalized by the principal according to $\pi(a_1,a_2)$ as in Definition~\ref{def:penalize-inconsistency}. To ensure that the agent is consistent in both stages with respect to his revealed type $\hat \theta = u_1/\phi'(a_1)$, we show that the adjustment function has to be an extreme form of penalty.


The following proposition shows that the only choice of penalty that guarantees
consistency is the extreme one with infinitely penalizes inconsistency.

\begin{prop}
Assume that $\pi$ penalizes inconsistency and that $a_1=a_{\hat \theta}(u_1)$ for some $\hat\theta \in\Theta$.
Then the choice of consistent second-stage action, $a_2=c(a_1)$, maximizes the agent's cumulative utility \eqref{eq:F_a-2} if and only if
\begin{equation}
    \label{eq:pi-crazy-infty}
    \pi(a_1,a_2) = \left\{\begin{array}{cl}
        0 & \text{if } a_2 = c(a_1) \\
        +\infty & \text{otherwise}.
    \end{array}\right.
\end{equation}
\end{prop}
\begin{proof}
    Let $a_1=a_{\hat \theta}(u_1)$ and $a_2=c(a_1)$. If $\pi$ is given by \eqref{eq:pi-crazy-infty}, we have
    \begin{equation}
        \label{eq:F-hat-a-consistency}
        \cva(a_1, a_2) = u_1 a_1 - \theta \phi(a_1) + u_2(a_1) a_2 - \theta \phi(a_2)
    \end{equation}
    and, for any $\hat a_2 \neq a_2$,
    \begin{equation}
        \label{eq:F-til-a-consistency}
        \cva(a_1, \hat a_2) = u_1 a_1 - \theta \phi(a_1) + u_2(a_1) \hat a_2 - \theta \phi(\hat a_2) - \pi(a_1,\hat a_2).
    \end{equation}
    \emph{Sufficiency.} Since $\pi(a_1,\hat
	a_2) = +\infty$ for $\hat a_2 \neq a_2$, we have $\cva(a_1,
	a_2) \geq \cva(a_1,\hat a_2)$ for every $\hat a_2\in A$.

    \emph{Necessity.} Assume $\cva(a_1,
	a_2) \geq \cva(a_1, \hat a_2)$ for every $\hat a_2\in A$, then from \eqref{eq:F-hat-a-consistency} and \eqref{eq:F-til-a-consistency}, we have
    \[
    \pi(a_1,\hat a_2) \geq \theta [\phi(a_2) - \phi(\hat a_2)] - u_2(a_1) (a_2 - \hat a_2).
    \]
    The above inequality must hold for every $\theta\in \R_+^*$, and neither $a_2$ nor $\hat a_2$ depend on $\theta$. Therefore, as $\theta\to \infty$, the only penalty satisfying the above inequality is \eqref{eq:pi-crazy-infty}. 
\end{proof}

\section{Concluding Remarks}

This paper examined truthful mechanism design in two-stage principal-agent contracts with adverse selection. By modeling it as a Stackelberg game under incomplete information, we characterized the constraints that information asymmetry imposed on incentive design.
We developed a two-stage contract structure that incorporated randomized first-stage incentives with contingent second-stage rewards. Our analysis demonstrated a fundamental distinction between continuous and discrete type spaces. For continuous types, we proved that truthful play could only be induced through constant second-stage incentives, as any reactive incentives would result in agents strategically misrepresenting their types through untruthful actions. In the discrete case, we established the existence of non-trivial truthful mechanisms through monotonic step functions, given that the step sizes are sufficiently small.
The introduction of an adjustment mechanism that penalized inconsistent behavior across stages expanded the space of feasible truthful contracts. This mechanism enabled more flexible incentive structures while maintaining truthful play. The framework developed in this paper laid groundwork for exploring complex dynamic contracting scenarios and mechanism design problems.

\bibliographystyle{abbrvnat}
\bibliography{refs}

\end{document}